\documentclass[letterpaper, 10 pt, conference]{ieeeconf}  
\usepackage{amsmath,amssymb,amsfonts}
\usepackage{graphicx}
\usepackage{textcomp}
\usepackage{comment}

\usepackage{amsthm}
\usepackage[utf8]{inputenc}
\usepackage[T1]{fontenc}
\usepackage{enumerate}
\usepackage[caption=false]{subfig}

\usepackage{soul}

\usepackage{hyperref}
\hypersetup{
    colorlinks=true,
    linkcolor=blue,
    filecolor=magenta,      
    urlcolor=cyan,
}

\usepackage[backend=biber,style=ieee,citestyle=ieee]{biblatex}
\AtBeginBibliography{\scriptsize}  
\newtheorem{theorem}{Theorem}

\newtheorem{corollary}{Corollary}

\theoremstyle{definition}
\newtheorem{definition}{Definition}

\newtheorem{remark}{Remark}

\usepackage{siunitx}
\DeclareSIUnit{\rad}{rad}

\usepackage[ruled,vlined]{algorithm2e}
\DontPrintSemicolon  

\usepackage{algpseudocode}

\IEEEoverridecommandlockouts                              

 \title{\LARGE \bf Barrier-Riccati Synthesis for Nonlinear Safe Control with Expanded Region of Attraction}

 \author{Hassan Almubarak \quad Maitham F. AL-Sunni \quad  Justin T. Dubbin \\ Nader Sadegh  \quad  John M. Dolan \quad   Evangelos A. Theodorou%
 \thanks{H. Almubarak is with the Control \& Instrumentation Engineering Department, King Fahd University of Petroleum and Minerals, Dhahran, Saudi Arabia. Email: {\tt{halmubarak@kfupm.edu.sa}}.
 M. F. AL-Sunni is with the Department of Electrical \& Computer Engineering, and J. M. Dolan is with the Robotics Institute, Carnegie Mellon University, Pittsburgh, PA,
USA. Emails: {\tt{\{malsunni, jdolan\}@andrew.cmu.edu}}.
J. T. Dubbin and E. A. Theodorou are with the Daniel Guggenheim School of Aerospace Engineering, and N. Sadegh is with the George W. Woodruff School of Mechanical Engineering, Georgia Tech, Atlanta, GA, USA. Emails: {\tt \{jdubbin3, evangelos.theodorou, sadegh\}@gatech.edu}.}%
 }

\begin{document}

\maketitle
\thispagestyle{empty}
\pagestyle{empty}

\begin{abstract}
We present a Riccati-based framework for safety-critical nonlinear control that integrates the barrier states (BaS) methodology with the State-Dependent Riccati Equation (SDRE) approach. The BaS formulation embeds safety constraints into the system dynamics via auxiliary states, enabling safety to be treated as a control objective. To overcome the limited region of attraction in linear BaS controllers, we extend the framework to nonlinear systems using SDRE synthesis applied to the barrier-augmented dynamics and derive a matrix inequality condition that certifies forward invariance of a large region of attraction and guarantees asymptotic safe stabilization. The resulting controller is computed online via point-wise Riccati solutions. We validate the method on an unstable constrained system and cluttered quadrotor navigation tasks, demonstrating improved constraint handling, scalability, and robustness near safety boundaries. This framework offers a principled and computationally tractable solution for synthesizing nonlinear safe feedback in safety-critical environments.
\end{abstract}

\vspace{-5mm}
\section{Introduction}
Safe autonomous control is increasingly demanded in modern simple and complex control systems. Since constraints are ubiquitous in control systems, different techniques have been proposed in the literature to handle them. Early in the control literature, many of those techniques were based on rendering set invariance through Lyapunov analysis \cite{blanchini1999set}. These methods were primarily developed for linear systems or for systems with simple constraints, such as box constraints, and did not generalize easily to nonlinear settings. However, with advancements in control theory and the aid of increased computational power, many of these approaches have been extended and further investigated for nonlinear systems. Barrier methods including barrier certificates \cite{prajna2003barrier,prajna2004safety}, control barrier functions (CBFs) \cite{wieland2007constructive,ames2014control,romdlony2014uniting}, Barrier-Lyapunov functions (BLFs) \cite{tee2009barrier_Lyap_automatica}, and most recently barrier states (BaS) \cite{almubarak2021safeddp,almubarak2023BaSTheorey}, have shown strong success in the recent safe control literature.

\subsection{Background and Related Work}
The barrier states (BaS) augmentation method enforces forward invariance of the set of allowed states through augmenting the \textit{state of safety} by means of barrier functions into the model of the safety-critical control system. This creates a transformed problem with the objective of seeking a stabilizing control law for the augmented model \cite{almubarak2023BaSTheorey}. Linear control methods, such as pole-placement and the linear quadratic regulator (LQR), were used to safely stabilize the safety-critical system \cite{Almubarak2021SafetyEC}. The same concept was extended to trajectory optimization in \cite{almubarak2021safeddp}, where BaS augmentation was shown to consistently outperform penalty methods and CBF safety filters in providing safe optimal trajectories. Additionally, extension of BaS to adaptive control methods was studied in \cite{aoun2023l1, al2025safety}.  However, a byproduct of this method is the increase in the nonlinearity of the system, including the conversion of linear models to nonlinear ones \cite{almubarak2023BaSTheorey}. Hence, the validity of the controller, portrayed through the region of attraction of the closed-loop system, is limited. This calls for a practical nonlinear control that provides a sufficiently large region of attraction.

State-Dependent Riccati Equation (SDRE) approaches are practical and widely used methods for obtaining sub-optimal control in nonlinear dynamical systems. They are used across many industries for regulation, control, observation, estimation, and filtering of nonlinear systems \cite{cimen2008state,nekoo2019SDREReview}. By factorizing the dynamics into state-dependent coefficient matrices, the nonlinear system is expressed in a form analogous to the standard linear state-space equation. A further advantage is that the formulation resembles a linear one, while the inherent nonlinearities of the original system remain present in the closed-loop dynamics. It is worth mentioning that the solution to the SDRE is not guaranteed to be optimal \cite{nekoo2019SDREReview}. Local optimality, however, is guaranteed under the standard assumptions that the parameterization is point-wise stabilizable, detectable, and continuously differentiable \cite{cloutier1996nonlinear}. 


Early efforts to integrate safety into SDRE-based control include introducing a fictitious output representing the rate of constraint violation, as in \cite{cloutier2001state}. The idea was to penalize this output in the cost, steering it to zero to maintain constraint satisfaction. However, this approach requires the fictitious output dynamics to be invertible, limiting its applicability to systems with an equal number of constraints and control inputs. A later method \cite{devi2021barrier}, inspired by BLFs, penalizes constraint violations through a BLF in the state-dependent weights in the cost. Yet, it is restricted to box constraints and introduces non-quadratic, non-differentiable terms, with a flawed stability proof based on an incorrect assumption about the time-derivative of the Riccati solution.

\subsection{Contributions and Organization}

The main contributions of this paper are threefold:
\begin{itemize}
    \item First, using the Mean Value Theorem with proper factorization, we derive the barrier state equation into an extended linearization form, resulting in safety embedded models that preserve the structure of the linear-like form needed in the SDRE technique. \item Second, we propose the safety-embedded State-Dependent Riccati Equation control framework as a nearly optimal, nonlinear control method with embedded safety guarantees. 
    \item Third, a major contribution is the derivation of a matrix inequality condition that is shown to provide semi-global asymptotic safe stability if satisfied. 
\end{itemize}

The paper is organized as follows. \autoref{sec: problem and preliminaries} presents the problem formulation. \autoref{sec: BaS-SDRE} introduces BaS, derives the extended BaS dynamics, and formulates the safety-embedded SDRE for approximating the constrained HJB solution, with theoretical guarantees for safe, near-optimal control. Conditions for extending local to semi-global safe stability are also given. Finally, simulation results are presented in \autoref{sec: implementation examples}, followed by conclusions in \autoref{sec: conslusions}.

\section{Safety-Critical Optimal Control Problem Formulation} \label{sec: problem and preliminaries}
\subsection{Safety-Critical Optimal Control Problem}
Consider the optimal control problem
\begin{align} \label{eq: cost functional}
     \min_{u\in\mathcal{U}}\frac{1}{2}\int_{0}^{\infty} \Big( \mathbf{Q}\big(x(t)\big) + u^{\top}(t)R\big(x(t)\big) u(t)\Big) dt 
\end{align}
subject to the nonlinear control-affine system 
\begin{align} \label{eq: control system dynamics}
    \dot{x}(t)= f(x(t))+ g(x(t)) u(t)
\end{align}
where $t\in \mathbb{R}$ is time\footnote{For notational convenience, dependence on time $t$ is dropped throughout the paper.}, $x \in \mathcal{X} \subset \mathbb{R}^n$ is the state, $u \in \mathcal{U} \subset \mathbb{R}^m$ is the control input, $x_0=x(0)$ is the initial state, $\mathbf{Q}: \mathbb{R}^n \rightarrow \mathbb{R}^+ \ \forall x \neq 0$ is a state cost function, $R: \mathbb{R}^n \rightarrow \mathbb{R}^{m \times m} \succ 0 \ \forall x \neq 0$ is a control cost function, $f: \mathbb{R}^n \rightarrow \mathbb{R}^n$ and $g: \mathbb{R}^n  \rightarrow \mathbb{R}^{n\times m}$ are at least $C^1$, and we assume that for the unforced system $f(0)=0$, without loss of generality, and that the linearized system is stabilizable and detectable.

The system \eqref{eq: control system dynamics} is subject to the safe set $\mathcal{S}$ defined by a continuously differentiable function $h(x)$, which is termed the safety function. Specifically, the safe set is defined as ${\mathcal{S}} = \left\{x \in \mathcal{X} | h(x) > 0  \right\}$ with $\partial{\mathcal{S}} = \left\{x \in \mathcal{X} | h(x) = 0 \right\}$ being the boundary set of $\mathcal{S}$. The set $\mathcal{S}$ must be rendered invariant under some feedback control to ensure the system's safety. 
\begin{definition}[\cite{blanchini1999set}] \label{def: controlled invariant set}
The set $\mathcal{S} \subset \mathcal{X}$ is controlled invariant for the nonlinear control system \eqref{eq: control system dynamics} if for any $x(0) \in \mathcal{S}$, there exists a continuous feedback controller $u={K}(x)$, such that the closed-loop system $\dot{x}= f(x)+ g(x) {K}(x)$ has the unique solution $x(t) \in \mathcal{S} \ \forall t \in \mathbb{R}^{+}$. Hence,
\begin{align} \label{eq: safety condition}
    h\big(x(t)\big) > 0 \ \forall t \in \mathbb{R}^{+} ; \ h\big(x(0)\big) > 0
\end{align}
which we will refer to as the safety condition throughout the paper. Therefore, the controller $u=K(x)$ is deemed safe.
\end{definition}
It is assumed that the system is safely stabilizable (\cite[Definition~4]{almubarak2023BaSTheorey}) in a set containing the origin, i.e., there exists a continuous controller $u(x)$, for which the origin of the closed-loop system $\dot{x}=f(x)+g(x)u(x)$ is asymptotically stable and that the safe set is controlled invariant. 


\section{Safety Embedded State Dependent Riccati Equation Control} \label{sec: BaS-SDRE}

Consider the optimal control problem \eqref{eq: cost functional}-\eqref{eq: control system dynamics}. In the linear–quadratic regulator (LQR) setting, the optimal value function is quadratic and obtained via the algebraic Riccati equation (ARE), yielding a linear state feedback law. The state-dependent Riccati equation (SDRE) method extends this idea to nonlinear systems by factorizing the dynamics into a pseudo-linear form with state-dependent matrices, then solving a Riccati equation point-wise. This enables locally Riccati-based, near-optimal feedback while preserving the structure and intuition of the linear theory.

Specifically, using \emph{extended linearization}, also known as \emph{state-dependent coefficient (SDC) parameterization}, we have
\begin{align}
\label{eq:extend_f_to_a}
    f(x) = A(x) x
\end{align}
where $A(x) \in \mathbb{R}^{n \times n}$ is a state-dependent matrix. This transforms the control-affine system \eqref{eq: control system dynamics} into the form
\begin{align} \label{eq: extended control system}
    \dot{x} = A(x) x + g(x) u
\end{align}
which mirrors the structure of a linear system while retaining the nonlinearities of the original dynamics.

The state cost in \eqref{eq: cost functional} can also be expressed in a quadratic form using the extended linearization, leading to the reformulated optimal control problem:
\begin{align} \label{eq: extended cost functional}
    \min_{u \in \mathcal{U}} \frac{1}{2} \int_0^\infty \left( x^\top Q(x)\,x + u^\top R(x)\,u \right) dt
\end{align}
subject to \eqref{eq: extended control system}, where $Q:\mathbb{R}^n \rightarrow \mathbb{R}^{n \times n} \succeq 0$ and $R:\mathbb{R}^n \rightarrow \mathbb{R}^{m \times m} \succ 0$ are at least $C^1$ and it is assumed that the pairs $\big(A(x),g(x)\big)$ and $\big(A(x),Q^{\frac{1}{2}}(x)\big)$ are point-wise stabilizable and detectable, respectively, $\forall x \in \Omega$ such that the solution to this problem, known as the value function $V(x)$, exists and is continuously differentiable. For a detailed treatment of the assumptions required for existence, uniqueness, stability, and near-optimality, see \cite{cimen2008state}.

\subsection{Factorized BaS and Extended Safety Embedded Systems}


For the nonlinear system \eqref{eq: control system dynamics} and the safe set $\mathcal{S}$ defined by the safety function $h(x)$, we define a barrier $\beta(x):=\mathbf{B}(h(x))$, where $\mathbf{B}$ is a classic barrier function, also known as an interior penalty function, that must satisfy the following definition.
\begin{definition} \label{def:bf}
    The function $\mathbf{B} : \mathbb{R} \to \mathbb{R}$ is a barrier function if it is smooth on $(0, \infty)$,  $\mathbf{B}(\eta) \xrightarrow[]{\eta \rightarrow 0} \infty$ and $\mathbf{B}' \circ \mathbf{B}^{-1} (\eta)$ is smooth as well\footnote{For notational convenience we use $\mathbf{B}'(\eta)$ to denote $\frac{d\mathbf{B}}{d\eta}(\eta)$.}. Examples for such barrier functions include the inverse barrier function $\mathbf{B}(\eta) = \frac{1}{\eta}$ and the logarithmic barrier function $\mathbf{B}(\eta) = -\text{log}\left(\frac{\eta}{1+\eta} \right)$. 
\end{definition}
Consequently, by definition, $\beta(x) \to \infty$ if and only if $h(x) \to 0$ (approaching unsafe regions) \cite{almubarak2023BaSTheorey}. The barrier's time derivative is then given by 
\begin{align}
\dot{\beta}(x) &= \mathbf{B}'\big(h(x) \big) L_{f(x)} h(x) + \mathbf{B}'\big(h(x) \big) L_{g(x)} h(x) u
\end{align}
which can be written as
\begin{align}
\dot{\beta}(x) &= \mathbf{B}'\big(\mathbf{B}^{-1}(\beta) \big) L_{f(x)} h(x) + \mathbf{B}'\big(\mathbf{B}^{-1}(\beta) \big) L_{g(x)} h(x) u \nonumber
\end{align}
Then, the \textit{barrier state}, denoted as $z$, is defined by the state equation
\begin{align}
\label{eq: barrier state equation}
\dot{z} &= \mathbf{B}'\big(\mathbf{B}^{-1}(z+\beta^{\circ}) \big) L_{f(x)} h(x) \\ &+ \mathbf{B}'\big(\mathbf{B}^{-1}(z+\beta^{\circ}) \big) L_{g(x)} h(x) u - \gamma \left(z+\beta^{\circ} - \beta(x) \right) \nonumber
\end{align}
where $\gamma \in \mathbb{R}^+$, $\beta^{\circ} =\beta(0)$ and the shift by $\beta^{\circ}$ ensures that the BaS dynamics vanish at $(x = 0,u = 0)$. It is worth noting that the term scaled by $\gamma$ ensures stabilizability of the new system after the BaS is augmented \cite{almubarak2023BaSTheorey}. Furthermore, it can easily be verified that the BaS dynamics \eqref{eq: barrier state equation} are smooth, given the smoothness of the system dynamics, $\beta(x)$, and the barrier function $\mathbf{B}$. For detailed derivations and theoretical discussions on BaS and related stability results, see \cite{almubarak2023BaSTheorey}.

By construction, the barrier state $z$ remains bounded if and only if the barrier function $\beta(x)$ remains bounded~\cite[Lemma~1]{almubarak2023BaSTheorey}; therefore, ensuring boundedness of $z$ directly guarantees system safety. The barrier state augmentation thus transforms the original system into an extended nonlinear system, even if the original system dynamics and constraints were linear. Consequently, a practical nonlinear control approach, such as the SDRE method, is particularly suitable. In what follows, we present an extended linearization formulation of this safety-embedded augmented system that fits naturally within the SDRE framework.

To incorporate the barrier states within the SDRE framework, we begin by rewriting the barrier function $\beta(x) = \mathbf{B}(h(x))$ in a factorized, linear-like form: 
\begin{equation}
     \beta(x) - \beta^{\circ}=\tilde{\beta}(x) x
\end{equation}
where $\tilde{\beta}(x) =  \int_0^1 \frac{\partial \beta}{\partial x} (\mu x) d\mu$ by the Mean Value Theorem. Using this factorization, the barrier state dynamics \eqref{eq: barrier state equation} can equivalently be expressed in an extended linearization form as
\begin{align} 
\label{eq: factorized barrier state equation}
 \dot{z} & =   \underbrace{\left(\mathbf{B}'\big(\mathbf{B}^{-1}(z+\beta^{\circ})\big)\frac{d h(x)}{d x} A(x) + \gamma \tilde{\beta}(x)\right)}_{A^z(x,z)} x  \\
 &\quad+ \underbrace{\mathbf{B}'\big(\mathbf{B}^{-1}(z+\beta^{\circ})\big)L_{g(x)} h(x)}_{g^z(x,z)} u - \gamma z  \nonumber
\end{align}
By augmenting the barrier state equation \eqref{eq: factorized barrier state equation} with the factorized dynamics \eqref{eq: extended control system}, we obtain the following extended safety embedded system:
\begin{align} 
\label{eq: safety embedded extended system}
    \dot{\bar{x}} = \bar{A}(\bar{x}) \bar{x} + \bar{g}(\bar{x}) u 
\end{align}
where the augmented state vector is defined as $\bar{x} = \begin{bmatrix} x \\ z\end{bmatrix}\in\bar{\mathcal{X}}\subseteq\mathcal{X}\times\mathcal{B}$, and the extended state-dependent matrices are
\begin{align}
\label{eq: Abar and Bbar}
    \bar{A}(\bar{x}) = 
\begin{bmatrix}
A(x) & 0_{n \times 1}\\[2pt]
A^z(x,z) & -\gamma
\end{bmatrix}, \quad
\bar{g}(\bar{x}) = 
\begin{bmatrix}
g(x) \\[2pt]
g^z(x,z)
\end{bmatrix}
\end{align}

We refer to \eqref{eq: safety embedded extended system} as the \textit{extended safety embedded dynamics}. The structure of the extended state-dependent matrices \eqref{eq: Abar and Bbar} is similar to the linear systems presented in \cite{almubarak2023BaSTheorey}. Here, we present a single BaS for simplicity. In general, for a system with multiple ($\mathcal{Q}$) constraints, one can either design a single aggregated barrier state using an aggregated safety function defined by
    $\frac{1}{h} = \sum_{i=1}^{\mathcal{Q}} \frac{1}{h_i}$,
or introduce multiple BaS separately, as shown in \cite{almubarak2021safeddp} or \cite{almubarak2023BaSTheorey}, depending on the problem and a desired performance. Note that this system has its origin as an equilibrium point: $\bar{A}(\bar{x})\bar{x}=0$, when $\bar{x}=0$. This preserves the smoothness and stabilizability of the original control system \eqref{eq: control system dynamics} \cite{Almubarak2021SafetyEC}. Therefore, the safety constraint is \textit{embedded} in the closed-loop system's dynamics, and stabilizing the safety embedded system \eqref{eq: safety embedded extended system} implies enforcing safety for the safety-critical system \eqref{eq: control system dynamics}, i.e., forward invariance of the safe set $\mathcal{S}$ with respect to \eqref{eq: control system dynamics}. This crucial result is formally established in \cite[Theorem~1]{almubarak2023BaSTheorey}, and a direct implication of the theorem was discussed in \cite[Remark~2]{almubarak2023BaSTheorey}. This result is formalized in the following corollary.
\begin{corollary}[{\cite[Theorem~1]{almubarak2023BaSTheorey}}] \label{theorem: original systems is safe if embedded system is stable}
Let $u=K(\bar{x})$ be a continuous feedback controller such that the origin of the safety embedded closed-loop system, $\dot{\bar{x}}= \bar{A}(\bar{x}) \bar{x} + \bar{g}(\bar{x}) K(\bar{x})$, is asymptotically stable. Then, there exists an open set $\mathcal{A}_{\text{safe}} \subseteq \mathcal{S}$ such that $u$ is safely stabilizing $\forall x(0) \in \mathcal{A}_{\text{safe}} \subseteq \mathcal{S}$. 
\end{corollary}

\subsection{Safety Embedded SDRE Control}
Leveraging the BaS technique, the safety-critical optimal control \eqref{eq: cost functional}-\eqref{eq: safety condition} is effectively transformed to the \textit{unconstrained} optimal control problem 
\begin{align} \label{eq: embedded cost functional}
   \min_{u \in \mathcal{U}} \frac{1}{2} \int_{0}^{\infty} \Big( \bar{x}^{\top} Q(\bar{x}) \bar{x} + u^{\top}R(\bar{x}) u\Big) dt 
\end{align}
subject to \eqref{eq: safety embedded extended system}, where $Q:\mathbb{R}^{n+q}\rightarrow \mathbb{R}^{{(n+q)}\times {(n+q)}} \succeq 0$, $R: \mathbb{R}^{n+q} \rightarrow \mathbb{R}^{m \times m} \succ 0 \ \forall \bar{x} \neq 0$, and $q$ is the number of constructed barrier states. Note that $Q$ needs to be chosen such that the pair $(\bar{A},Q^{\frac{1}{2}})$ is detectable after adding the BaS \cite{cimen2008state}. It is worth noting that with this formulation, the $Q$ matrix can be a simple real-valued constant matrix which ensures the convexity of the cost function. That is, the BaS can be penalized as any other state and hence the proposed method provides a natural way to incorporate constraints. By \autoref{theorem: original systems is safe if embedded system is stable}, solving this optimal control problem implies solving both the safety and stability problems. 

As discussed earlier, the quality of the solution also depends on the extended linearization of $f(x)$. It is worth mentioning that, other than evaluating $\tilde{\beta}$, the extended linearization of the safety embedded system completely depends on the factorization of the original dynamics. For such a problem, following the SDRE approach, the safety embedded sub-optimal control law is given by 
\begin{align} \label{eq: suboptimal safety embedded control law}
    u_{\rm{safe}}(\bar{x}) = - R^{-1}(\bar{x}) \bar{g}^{\top}(\bar{x}) P(\bar{x})\bar{x} 
\end{align}
and the associated safety embedded gradient of the value function is 
\begin{align} \label{eq: embedded value function}
\frac{\partial V}{\partial \bar{x}}^\top = P(\bar{x}) \bar{x}, \ P:\mathbb{R}^{n+q}\rightarrow \mathbb{R}^{{n+q} \times {n+q}}
\end{align}
where $P(\bar{x}) \succ 0$ solves the \textit{safety embedded state-dependent Riccati equation}, that we refer to as BaS-SDRE, which is given by
\begin{align} \label{eq: safety embedded state-dependent Riccati equation}
P(\bar{x})\bar{A}(\bar{x}) + \bar{A}^{\top}(\bar{x})P(\bar{x})- P(\bar{x}) \bar{G}(\bar{x}) P(\bar{x}) + Q(\bar{x}) = 0
\end{align} 
where $\bar{G}(\bar{x}):=\bar{g}(\bar{x}) R(\bar{x})^{-1} \bar{g}(\bar{x})^\top$. This control law results in the closed-loop safety embedded system 
\begin{align} \label{eq: closed-loop ses under bas-sdre}
    \dot{\bar{x}} = \bar{A}_c(\bar{x})\bar{x}
    \end{align}
    where
    \begin{align} \label{eq: closed-loop se_matrix under sdre}
        \bar{A}_c (\bar{x})= \bar{A}(\bar{x}) - \bar{g}(\bar{x})R^{-1}(\bar{x})\bar{g}^\top(\bar{x})P(\bar{x})
  \end{align}
Similar to the standard SDRE formulation, we assume that the pairs $\big(\bar{A}(\bar{x}),\bar{g}(\bar{x})\big)$ and
$\big(\bar{A}(\bar{x}),Q^{\frac{1}{2}}(\bar{x})\big)$ are point-wise stabilizable and detectable, respectively, $\forall \bar{x} \in \bar{\Omega} = \Omega \times \mathcal{B}$. Nonetheless, it is crucial now to note that those entities, and most importantly the feedback control law, are functions of the barrier states, which provide the controller with the safety status of the system. The controller, in turn, reacts to this status, as we shall see in the numerical simulations.

\subsection{Local Safe Stability and ROA Expansion Condition}
By construction, the new safety-embedded control solution via the BaS-SDRE follows the regular SDRE conditions in optimality and stability, albeit one must note that the state space of operation is now $\bar{\mathcal{X}}$. That is, the solution of the BaS-SDRE depends on the factorization of the system and the barrier that is not unique, which may lead to factorization-dependent control laws that provide different stability, safety, and optimality performances. \textcite{ccimen2010systematic} proposed a systematic method to define systems' and costs' matrices, yet a general rule for defining the \textit{optimal} factorization is still open research. In addition, \cite[Theorem~2]{cimen2008state} has shown the necessary conditions of local asymptotic stability of SDRE-based control. This was done by applying the Mean Value Theorem to $\bar{A}_c(\bar{x})$ giving $\dot{\bar{x}} =
\bar{A}_\mathrm{c}({0}) \bar{x}+\mathcal{A}(\bar{x}) \bar{x}$ for some $\mathcal{A}(\bar{x})$ that satisfies $||\mathcal{A}(\bar{x})|| =\mathcal{O}\left(||\bar{x}||\right)$. Since $\bar{A}_\mathrm{c}(0)$ is Hurwitz (due to point-wise stabilizability at~$\bar{x}=0$), and 
$\mathcal{A}(\bar{x})\bar{x}=\mathcal{O}(\|\bar{x}\|^2)$ vanishes faster than the linear term in $\bar{\Omega} = \Omega \times \mathcal{B}$, it follows that the origin is locally asymptotically stable \cite{cloutier1996nonlinear,cimen2008state}. This result naturally extends to safety embedded systems, given the preservation of smoothness and stabilizability \cite{almubarak2023BaSTheorey}, and by \autoref{theorem: original systems is safe if embedded system is stable}, the local safe stability of the original system is achieved. 

Nevertheless, local guarantees may not be sufficient and could be limited. For safety-critical systems, it is desirable to maximize the forward invariant region of the closed-loop system and similarly maximize the region of attraction (ROA) of the safe stabilization problem \cite{Wang2018Permissive}. In what follows, using the regularity conditions needed for the SDRE approach, we derive a condition based on the problem matrices, namely the cost function design, when one can expand asymptotic stability to larger sets, potentially achieving semi-global asymptotic stability on an arbitrarily large set $\Omega_c$. Note that the first part of the proposition in the theorem only states the SDRE regularity assumptions.

\begin{theorem} \label{thm:BaS-SDRE}
 Consider the safety embedded optimal control problem \eqref{eq: embedded cost functional} subject to the safety embedded system  \eqref{eq: safety embedded extended system} with $\bar{A}(\bar{x}), \bar{g}(\bar{x}), Q(\bar{x})$, and $R(\bar{x})$ are at least $C^1$, and the pairs $\big(\bar{A}(\bar{x}),\bar{g}(\bar{x})\big)$ and $\big(\bar{A}(\bar{x}),Q^{\frac{1}{2}}(\bar{x})\big)$ are point-wise stabilizable and point-wise detectable respectively $\forall \bar{x} \in \bar{\Omega}$ such that the point-wise solution $P(\bar{x})$ to the BaS-SDRE \eqref{eq: safety embedded state-dependent Riccati equation} renders $\bar{A}_c(\bar{x})$ point-wise Hurwitz. The safety-critical system \eqref{eq: control system dynamics} under the safety embedded nonlinear feedback control law \eqref{eq: suboptimal safety embedded control law} is (semi-globally) asymptotically \textbf{safely} stable with ROA $\Omega_c = \{x\in \Omega, \ \bar{x}^\top P(\bar{x}) \bar{x} \leq c\}$ for some $c >0$, if 
     \begin{align}
        Q(\bar{x}) - \dot{P}(\bar{x}) \succ 0
    \end{align}
 where $\dot{P}$ is the time derivative of the BaS-SDRE solution $P$ along the trajectories of \eqref{eq: closed-loop ses under bas-sdre}, which can be explicitly computed as specified in \eqref{eq: P_bar_dot}.
\end{theorem}
\begin{proof}
    Consider the candidate Lyapunov function
    \begin{align}
       W(\bar{x}) = \bar{x}^\top P(\bar{x}) \bar{x}
    \end{align}
    It can be seen that the time derivative of $W(\bar{x})$ along the trajectories of the closed-loop system is given by
     \begin{align}
     \dot{W} & = \bar{x}^{\top}\Big( P(\bar{x}) \bar{A}_c({\bar{x}}) +  \bar{A}^{\top}_c({\bar{x}})  {P}({\bar{x}}) \Big) {\bar{x}} + \bar{x}^\top \dot{P}(\bar{x}) \bar{x} \\
     & = {\bar{x}}^{\top}\Big(- {Q}({\bar{x}}) - {P}({\bar{x}}) \bar{G}(\bar{x}) {P}({\bar{x}})\Big) {\bar{x}} + \bar{x}^\top \dot{P}(\bar{x}) \bar{x} \nonumber
     \end{align}
    Hence, 
    \begin{equation} 
    \label{eq: Derivative Candidate}
        \dot{W}(\bar{x}) = \bar{x}^\top \left(-\hat{Q}(\bar{x})- P(\bar{x})\bar{G}(\bar{x})P(\bar{x}) \right) \bar{x} \leq -\bar{x}^\top \hat{Q}(\bar{x}) \bar{x}
    \end{equation}
    where $\hat{Q}(\bar{x})=Q(\bar{x})-\dot{P}(\bar{x})$. 
    
    Then, by the hypothesis, $\dot{W}(\bar{x}) \leq -\bar{x}^\top \hat{Q}(\bar{x}) \bar{x} < 0$ for all $\bar{x} \in \bar{\Omega}_c= \{\bar{x}\in \bar{\Omega}, \ \bar{x}^\top P(\bar{x}) \bar{x} \leq c\}$. Finally, by Lyapunov theory $\bar{\Omega}_c$ is positively invariant \cite{khalil2002nonlinear}, implying that $\bar{\Omega}_c$ is indeed an ROA for the closed-loop system. It remains to show how $\dot{P}$ can be explicitly computed. 
    
    Given that under the regularity assumptions, the value function is smooth, $P(\bar{x})$ is differentiable. Hence, as in standard Lyapunov analyses, the time derivative $\dot{P}$ along the trajectory of the system of the solution $P$ of the SDRE can be determined by differentiating \eqref{eq: safety embedded state-dependent Riccati equation}:
    \begin{align} \label{eq: d/dt SDRE}
        \dot{P}(\bar{x})\bar{A}_c(\bar{x}) + \bar{A}_c^{\top}(\bar{x})\dot{P}(\bar{x})+\tilde{Q}(\bar{x}) =0
    \end{align}
    where 
$       \tilde{Q}(\bar{x})=P(\bar{x})\dot{\bar{A}}(\bar{x})+\dot{\bar{A}}^{\top}(\bar{x})P(\bar{x})- P(\bar{x}) \dot{\bar{G}}(\bar{x}) P(\bar{x}) + \dot{Q}(\bar{x})$
    and the time derivatives of the elements of $\bar{A}$, $\bar{G}$ and $Q$ are computed as follows:
    \small
    \begin{equation*}
       \dot{\bar{A}}_{ij}=\frac{\partial \bar{A}_{ij}}{\partial \bar{x}} \bar{A}_c(\bar{x})\bar{x}, \ \dot{\bar{G}}_{ij}=\frac{\partial \bar{G}_{ij}}{\partial \bar{x}} \bar{A}_c(\bar{x})\bar{x}, \ \ \dot{Q}_{ij}=\frac{\partial Q_{ij}}{\partial \bar{x}} \bar{A}_c(\bar{x})\bar{x}
    \end{equation*}
    \normalfont
    The solution of the Lyapunov equation \eqref{eq: d/dt SDRE} is given by $\dot{P}(\bar{x})$
    which satisfies
    \begin{align} \label{eq: P_bar_dot}
    \text{vec}\big(\dot{P}(\bar{x})\big)=-(\bar{A}_c^\top(\bar{x})\otimes I+I\otimes \bar{A}_c^\top(\bar{x}))^{-1} \text{vec}\big(\tilde{Q}(\bar{x})\big) 
    \end{align}
    where $\otimes$ is the Kronecker product and $\text{vec}(\cdot)$ is the vectorization operator. Therefore, the safety embedded system is (semi-globally) asymptotically stable on $\bar{\Omega}_c = \{\bar{x} \in \bar{\Omega}, \ \bar{x}^\top P(\bar{x}) \bar{x} \leq c\}$. 
    
    Finally, let $\Omega_c \times \mathcal{B}\subset \bar{\Omega}_c$. 
    By \autoref{theorem: original systems is safe if embedded system is stable}, the nonlinear feedback controller safely stabilizes the original system \eqref{eq: control system dynamics} (semi-globally) with $\Omega_c$ as its positively invariant ROA.
\end{proof}

\begin{remark}
It is important to note that the state-feedback control law \eqref{eq: suboptimal safety embedded control law} derived via the SDRE is \emph{suboptimal} in the sense that it enforces a quadratic structure on the value function and satisfies the HJB equation only point-wise. Consequently, while the controller guarantees asymptotic stability on $\Omega_c$, it does not, in general, solve the full nonlinear \textit{optimal} control problem. \autoref{thm:BaS-SDRE} shows that $Q(\bar{x})$ is chosen to ensure that $Q(\bar{x}) - \dot{P} \succ 0$ for all $\bar{x} \in \bar{\Omega}_c$ for the largest possible $c>0$, and consequently the largest possible ROA, $\bar{\Omega}_c$. It is worth noting that with an appropriate choice of $Q(\bar{x})$ (e.g., by fixing $R(\bar{x})$), $\bar{\Omega}_c$ can be made arbitrarily large, yielding a semi-global asymptotic stability of the origin of the safety embedded system.
\end{remark}

\section{Numerical Simulations} \label{sec: implementation examples}
We validate our method on two systems: a constrained 2D linear toy system and a planar quadrotor moving in environments full of obstacles. Our approach is compared against the \textit{BaS-LQR} approach, which linearizes the dynamics to solve the infinite horizon optimal control problem. In all experiments, we set $\gamma = 1$ for the barrier states and choose the barrier function as $\mathbf{B}(\eta) = 1/\eta$. 

\subsection{Unstable Constrained Linear System}
Consider the unstable system
\begin{align}
    \begin{bmatrix} \dot{x}_1 \\ \dot{x}_2 \end{bmatrix}= \begin{bmatrix} 1 & -5 \\ 0 & -1\end{bmatrix} \begin{bmatrix} x_1 \\ x_2 \end{bmatrix} + \begin{bmatrix} 0 \\ 1\end{bmatrix} u
\end{align}
We define the safety constraint  as $(x_1 - 2)^2 + (x_2 - 2)^2 \geq 0.5^2 $.
Following the formulation presented earlier, the barrier state equation is given by
\begin{equation*} \begin{split}
    \dot{z} = -\left(z + \frac{4}{31}\right)^2 &  \left[2(x_1-2) \ \ 2(x_2-2)\right] \Bigg(\begin{bmatrix} 1 & -5 \\ 0 & -1\end{bmatrix} \begin{bmatrix} x_1 \\ x_2 \end{bmatrix} \\+ 
    & \begin{bmatrix} 0 \\ 1\end{bmatrix} u \Bigg) - \left(z + \frac{4}{31} - \beta(x) \right)
\end{split} \end{equation*}
Following the proposed technique, we rewrite the safety-embedded system as
\begin{equation*}
    \dot{\bar{x}} = \bar{A}(\bar{x}) \bar{x} + \bar{g}(\bar{x}) u
\end{equation*}
where 
$\bar{A}(\bar{x}) = 
\begin{bmatrix}
1 & -5 & 0 \\ 0 & -1 & 0 \\ A_1^z(x,z) & A_2^z(x,z) & -1
\end{bmatrix} , \
\bar{g}(\bar{x}) = 
\begin{bmatrix}
0 \\ 1 \\ g^z
\end{bmatrix}$ 
with $A_1^z(x,z) = \frac{1}{(2x_1((x_1 - 2)^2 + (x_2 - 2)^2 - 1/4) + 100)} - (2x_1 - 4)(z + \frac{4}{31})^2$, $A_2^z(x,z) = \frac{1}{(2x_2((x_1 - 2)^2 + (x_2 - 2)^2 - 1/4) + 100)} + 10(x_1 - 2)(z + \frac{4}{31})^2+ 2(x_2 - 2)(z + \frac{4}{31})^2$ and $g^z = -2(x_2 - 2)(z + \frac{4}{31})^2$.

\autoref{fig: linear toy state space LQR vs SDRE comparison} shows trajectories resulting from numerical simulations of the closed-loop system starting from random initial conditions for the different approaches. It can be seen that the nonlinear controller, BaS-SDRE, can safely stabilize the system for all the different initial conditions at which the linear controller, BaS-LQR, fails.

\begin{figure}[t]
\centering
\vspace{-4mm}
    {\includegraphics[trim=0 0 0 0, clip, width=.8\linewidth]{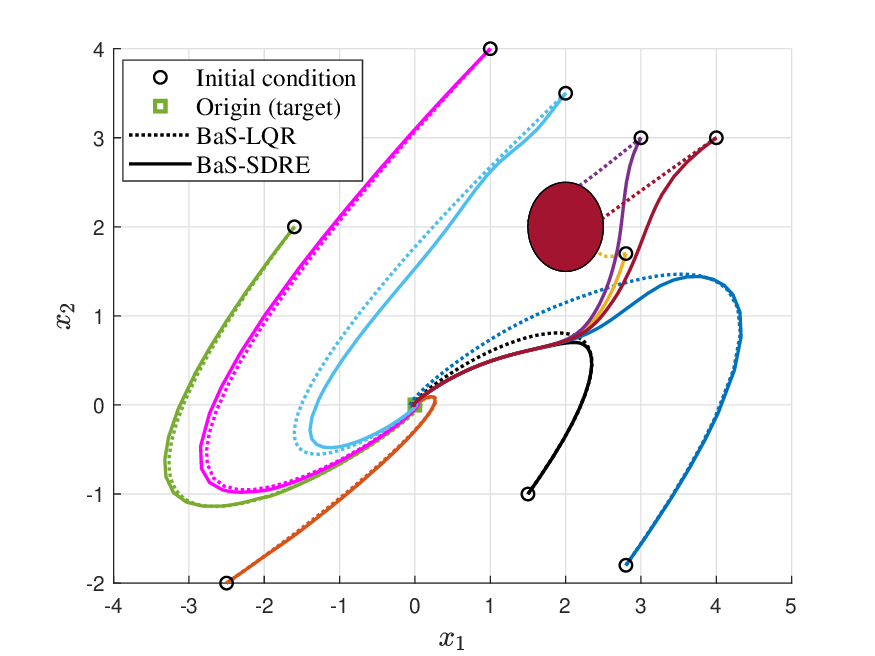}}
      \caption{Simulations of the closed-loop system under the BaS-LQR controller (dotted) and under the proposed BaS-SDRE controller (solid) starting from different initial conditions (black circles) with the unsafe region shown as a dark red circle.}
      \label{fig: linear toy state space LQR vs SDRE comparison}
      \vspace{-5mm}
\end{figure}  

To see how the nonlinearity stemming from the constraint (and specifically the barrier state in the safety embedded approach) affects the controller, \autoref{fig: linear toy far i.c. with gains and BaS} shows a simulation example using an initial condition that is relatively far from the origin. Interestingly, it can be seen that the SDRE solution (the gain) follows the linear solution when the system is far from the unsafe region, and it departs from it as the system approaches the unsafe region due to the nonlinearity coming from the barrier state. This can be seen in the plot of the barrier state, over time, where there is a clear peak when the obstacle is very close.

\begin{figure*}[th]
\centering
    {\includegraphics[trim=0 0 0 0, clip, width=.3\linewidth]{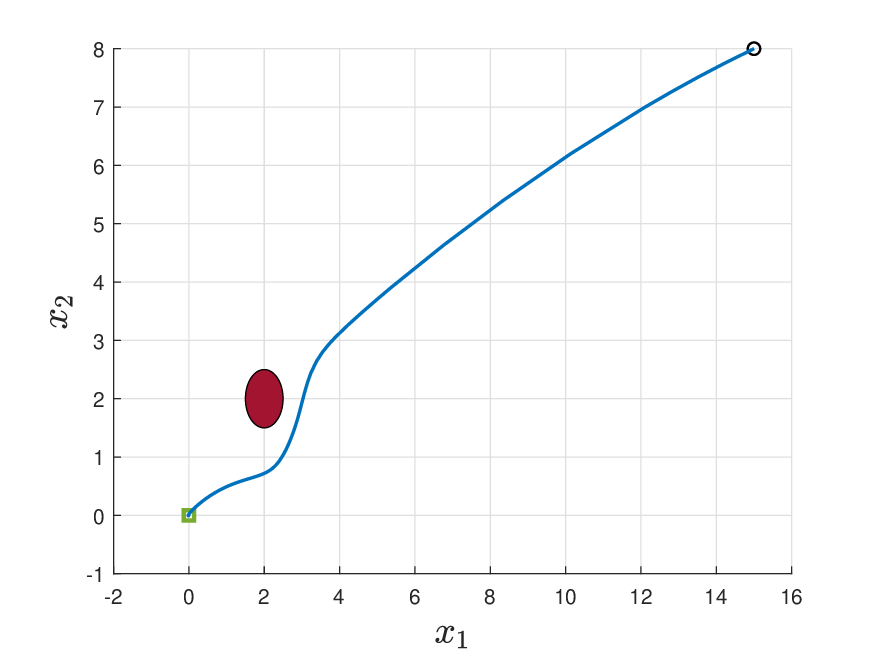}}
    {\includegraphics[trim=0 0 0 0, clip, width=.3\linewidth]{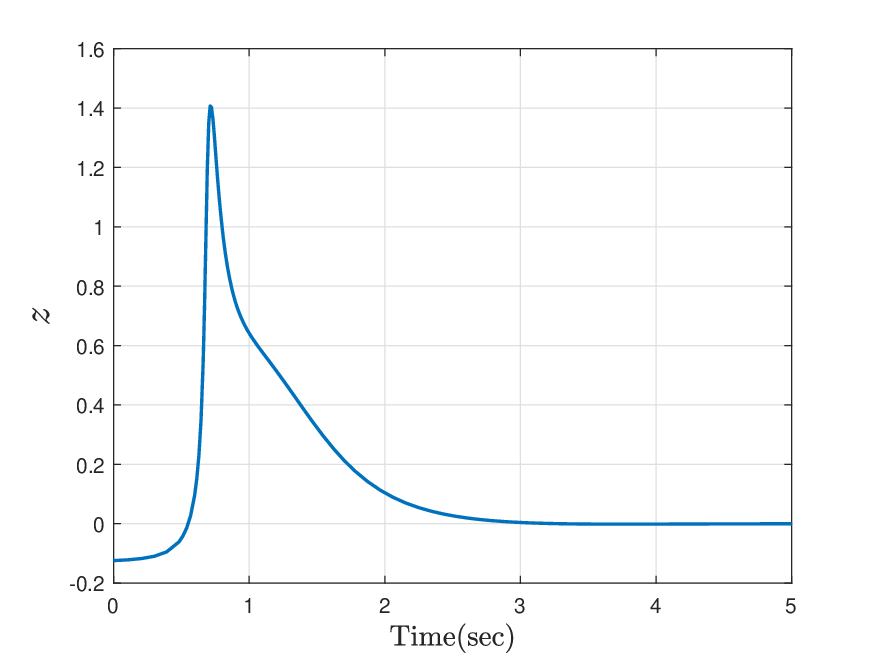}}
    {\includegraphics[trim=0 0 0 0, clip, width=.3\linewidth]{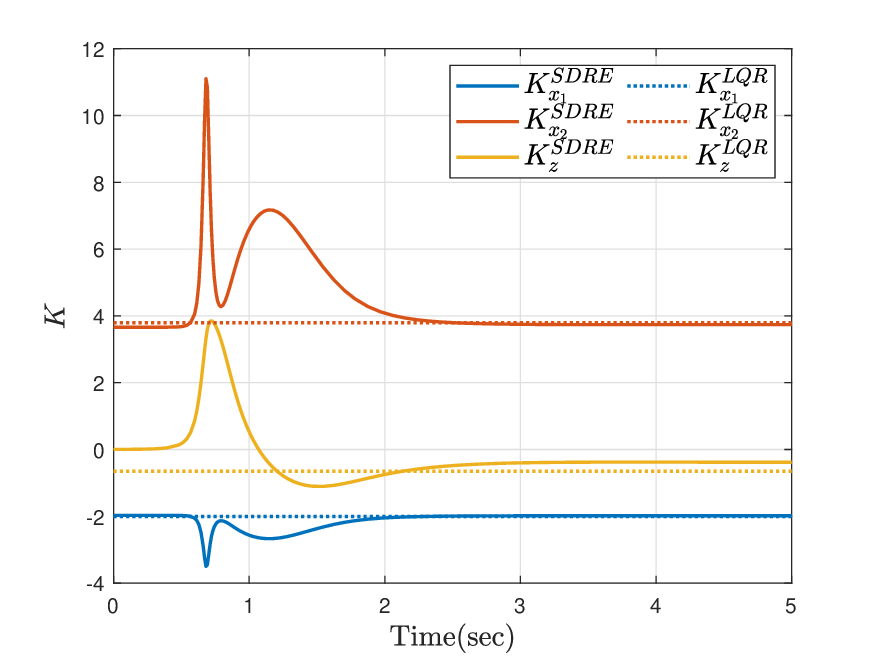}}
      \caption{Simulation results show the system trajectory, barrier state evolution, and feedback gains under the BaS-SDRE controller. As the barrier state grows, the BaS-SDRE adapts its gains to capture the system’s nonlinear behavior, differing from the fixed LQR gains.}
      \label{fig: linear toy far i.c. with gains and BaS}
      \vspace{-5mm}
\end{figure*}

\subsection{Planar quadrotor}
Consider the quadrotor dynamics given by
\begin{align}
    \begin{bmatrix}
        \Ddot{x} \\ 
        \Ddot{y} \\
        \Ddot{\psi}
    \end{bmatrix}  = 
    \begin{bmatrix}
        \frac{1}{m} (u_1 + u_2) \sin{(\psi)}  \\ 
        \frac{1}{m} (u_1 + u_2) \cos{(\psi)} - g  \\
        \frac{l}{2J} (u_2 - u_1)   \\
    \end{bmatrix}
\end{align}
where ($x,y$) is the location, $\psi$ is the orientation, $m$ is the mass of the quadrotor, $l$ is the distance from the center to the propellers of the quadrotor, $J$ is the moment of inertia, $g$ is the acceleration due to gravity, and $u_1$ and $u_2$ denote the right and left thrust forces respectively. The objective is to safely thrust the quadrotor to hover from a random initial condition where the quadrotor might flip. Specifically, we want to safely stabilize the quadrotor at $x=\SI[per-mode = symbol]{0}{\meter}, y=\SI[per-mode = symbol]{0}{\meter}, \psi = \SI[per-mode = symbol]{0}{\rad}$, where there are multiple circular obstacles of different sizes on the way. For the example, we use the following parameters: $m=\SI[per-mode = symbol]{1}{\kilogram}$, $l=\SI[per-mode = symbol]{0.3}{\meter}$, $g = \SI[per-mode = symbol]{9.81}{\meter\per\square\second}$, and $J = 0.2 \cdot m \cdot l^2$. 

\autoref{fig:quadrotor_sim} shows two different scenarios for a quadrotor navigating in a crowded environment. Our approach safely completes the task, while vanilla SDRE (without considering safety specifications) and BaS-LQR fail to ensure the safety of the quadrotor in both scenarios, the case where we have spaces between the different obstacles, and the case when a lot of obstacles are close to each other, making it hard for the quadrotor to pass in between them

\begin{figure}[th]
    \centering
    \includegraphics[width=.7\linewidth]{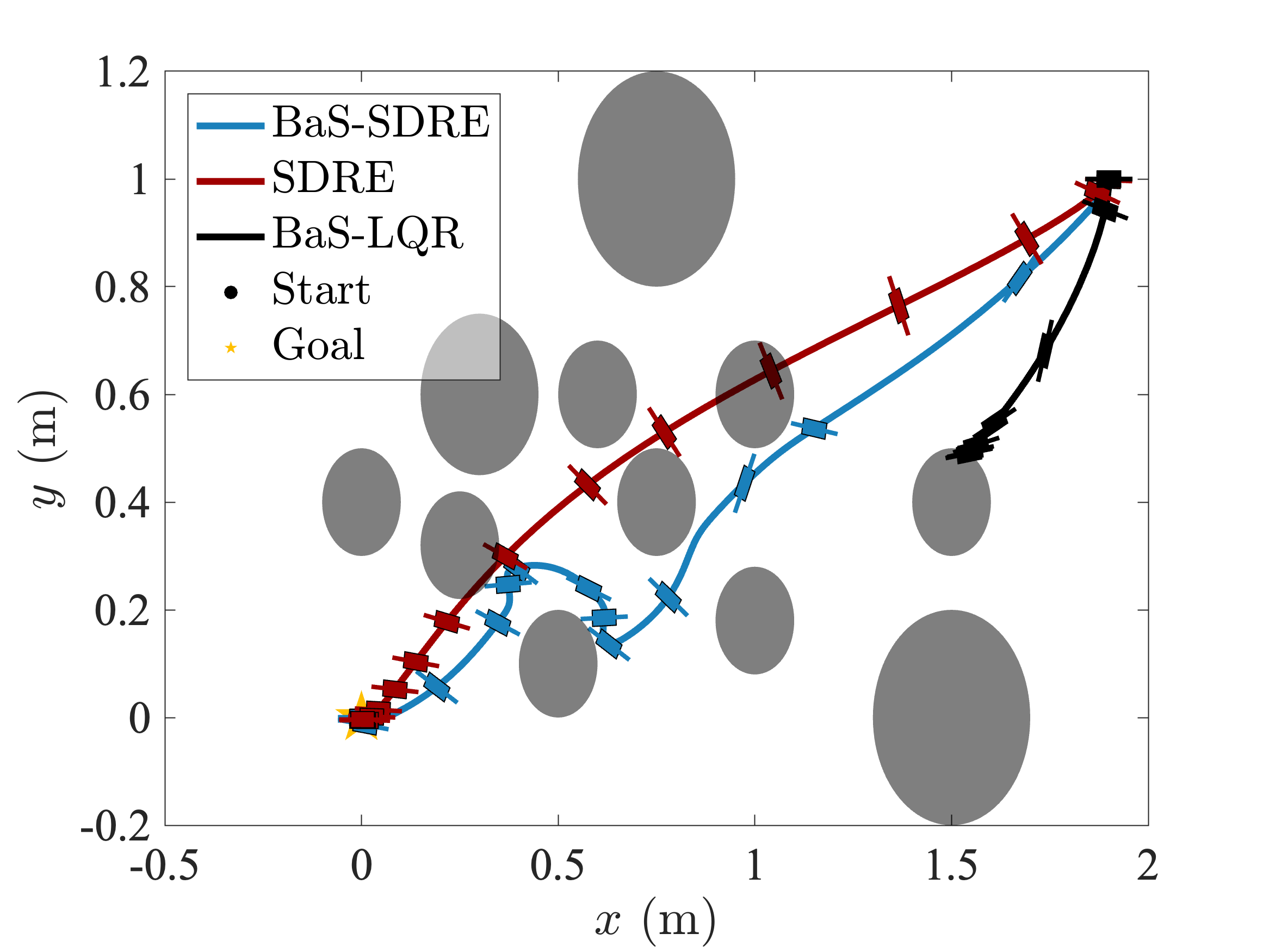} \\
    \includegraphics[width=.7\linewidth]{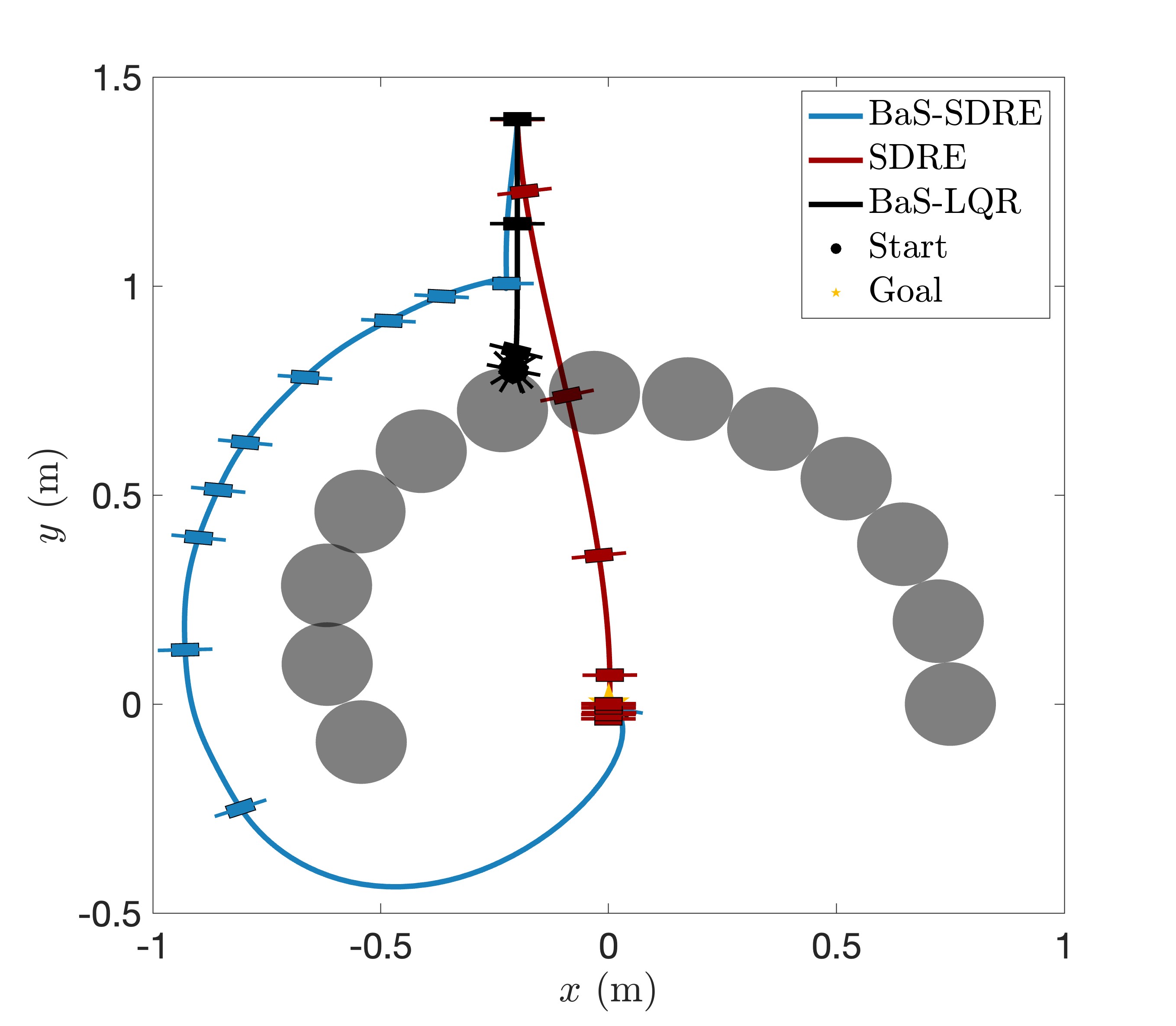}
    \caption{A planar quadrotor navigating in two different obstacle courses. Our method successfully stabilizes the quadrotor at the target position, while vanilla SDRE and BaS-LQR cannot, having their trajectories crossing unsafe regions.}
    \label{fig:quadrotor_sim}
    \vspace{-7.5mm}
\end{figure}

\section{Conclusions}  \label{sec: conslusions}
This paper presented BaS-SDRE, a unified framework for synthesizing safe nonlinear feedback by embedding barrier dynamics into a Riccati-based structure. By extending the state-dependent Riccati equation formulation to explicitly incorporate barrier-state augmentation, the method enables control design that jointly addresses performance and safety. The proposed approach preserves the pseudo-linear structure required for Riccati synthesis, allowing efficient computation of nonlinear state feedback while ensuring that safety constraints are encoded dynamically within the system evolution. We established that the closed-loop system achieves asymptotic safe stabilization, and derived a matrix inequality condition that certifies forward invariance of an enlarged region of attraction under mild assumptions. Through simulations on constrained and underactuated systems, BaS-SDRE demonstrated improved constraint handling and tracking performance relative to baseline SDRE and BaS-LQR controllers. These results suggest that BaS-SDRE provides a principled and scalable tool for nonlinear safe control, with future extensions targeting robustness to uncertainty, discrete-time implementations, and integration with learning-based adaptation.

\vspace{-2mm}
\printbibliography
\end{document}